\documentclass[11pt]{article}

\usepackage[letterpaper,margin=1in]{geometry}
\usepackage{amsmath,amssymb,amsthm,mathtools}
\usepackage{aliascnt}
\usepackage[ruled,vlined,linesnumbered]{algorithm2e}
\usepackage{booktabs}
\usepackage{enumitem}
\usepackage{etoolbox}
\usepackage{microtype}
\usepackage{xcolor}
\usepackage{placeins}
\usepackage[colorlinks=true,linkcolor=blue!55!black,citecolor=blue!55!black,
urlcolor=blue!55!black]{hyperref}
\usepackage[nameinlink,noabbrev,capitalize]{cleveref}

\allowdisplaybreaks
\setlist[itemize]{leftmargin=1.7em}
\setlist[enumerate]{leftmargin=2.1em}

\newtheorem{theorem}{Theorem}[section]
\newaliascnt{lemma}{theorem}
\newtheorem{lemma}[lemma]{Lemma}
\aliascntresetthe{lemma}
\newaliascnt{corollary}{theorem}
\newtheorem{corollary}[corollary]{Corollary}
\aliascntresetthe{corollary}
\newaliascnt{proposition}{theorem}

\aliascntresetthe{proposition}
\theoremstyle{definition}
\newaliascnt{definition}{theorem}
\newtheorem{definition}[definition]{Definition}
\aliascntresetthe{definition}
\theoremstyle{remark}
\newaliascnt{remark}{theorem}
\newtheorem{remark}[remark]{Remark}
\aliascntresetthe{remark}

\crefname{algocf}{Algorithm}{Algorithms}
\Crefname{algocf}{Algorithm}{Algorithms}
\crefname{theorem}{Theorem}{Theorems}
\Crefname{theorem}{Theorem}{Theorems}
\crefname{lemma}{Lemma}{Lemmas}
\Crefname{lemma}{Lemma}{Lemmas}
\crefname{corollary}{Corollary}{Corollaries}
\Crefname{corollary}{Corollary}{Corollaries}
\crefname{proposition}{Proposition}{Propositions}
\Crefname{proposition}{Proposition}{Propositions}
\crefname{definition}{Definition}{Definitions}
\Crefname{definition}{Definition}{Definitions}
\crefname{remark}{Remark}{Remarks}
\Crefname{remark}{Remark}{Remarks}
\crefname{equation}{Equation}{Equations}

\newcommand{\M}{\mathcal{M}}
\newcommand{\I}{\mathcal{I}}
\newcommand{\nS}{\overline{S}}
\newcommand{\Set}[1]{\lbrace #1 \rbrace}
\newcommand{\Aact}{A_{\mathrm{act}}}
\newcommand{\Bact}{B_{\mathrm{act}}}

\SetKwInput{KwInput}{Input}
\SetKwInput{KwOutput}{Output}

\SetKwFunction{HasIn}{HasIn}
\SetKwFunction{HasOut}{HasOut}
\SetKwFunction{FindIn}{FindIn}
\SetKwFunction{FindOut}{FindOut}
\SetKwFunction{ListIn}{ListIn}
\SetKwFunction{ListOut}{ListOut}
\SetKwFunction{CollectOut}{CollectOut}
\SetKwFunction{ClassifyColumns}{ClassifyColumns}
\SetKwFunction{Emit}{Emit}
\SetKwFunction{MinLoadNeighbor}{MinLoadNeighbor}
\SetKwFunction{AssignWitness}{AssignWitness}
\SetKwFunction{InitializeZeroRows}{InitializeZeroRows}
\SetKwFunction{DeleteColumn}{DeleteColumn}
\SetKwFunction{DeactivateRow}{DeactivateRow}
\SetKwFunction{RepairColumns}{RepairColumns}
\SetKwFunction{InitializeZeroColumns}{InitializeZeroColumns}
\SetKwFunction{DeleteRow}{DeleteRow}
\SetKwFunction{DeactivateColumn}{DeactivateColumn}
\SetKwFunction{Another}{Another}
\SetKwFunction{Enumerate}{Enumerate}
\SetKwFunction{RecEnumerate}{RecEnumerate}

\SetKwProg{Function}{Function}{:}{}
\SetKwProg{Procedure}{Procedure}{:}{}
\SetKw{Return}{return}
\SetKw{Continue}{continue}
\SetKw{Break}{break}
\SetKwComment{Comment}{$\triangleright$\ }{}

\title{Subquadratic-Query Algorithms for Finding\\
Another Maximum Matroid Intersection}
\newcommand{\AuthorAffiliation}{Institute of Science Tokyo}
\newcommand{\AuthorEmail}{}
\author{Makoto Watanabe\\
    \small \AuthorAffiliation%
    \ifdefempty{\AuthorEmail}{}{\\[-0.2ex]
        \small\href{mailto:\AuthorEmail}{\texttt{\AuthorEmail}}}%
}
\date{}

\begin{document}
\maketitle

\begin{abstract}
Let $\M_1,\M_2$ be two matroids on a common ground set $V$, given by
independence oracles, and let $S$ be a maximum-cardinality common independent set.
We study the problem of deciding whether there exists another maximum common
independent set $T\ne S$, and of outputting one when it exists.  Writing
$n\coloneqq|V|$ and $r\coloneqq|S|$, we give a Las Vegas algorithm using
$\widetilde O(n\sqrt r)$ expected independence queries and a deterministic
algorithm using $\widetilde O(nr^{2/3})$ queries.  As an application, all
maximum common independent sets can be
enumerated with at most two another-solution calls between consecutive outputs
and after the last output; if there are $L$ solutions, exactly $2L-1$ such
calls are made.  Neither algorithm constructs the exchange graph.
Instead, they perform Kahn-style source peeling through two deletion-only data
structures.  One side uses the heavy/light categorization of Blikstad, van den
Brand, Mukhopadhyay, and Nanongkai.  For the opposite side, where no transposed
oracle is available, we give a collective randomized classifier and a
deterministic capacity-saturation certificate.
\end{abstract}

\section{Introduction}

\paragraph{Problem and motivation.}
Matroid Intersection is a fundamental problem in combinatorial optimization.
Given two matroids $\M_1=(V,\I_1)$ and $\M_2=(V,\I_2)$ on the same ground
set, it asks for a largest set that is independent in both matroids.  We work
in the standard independence-oracle model: the only access to each matroid is
an oracle that decides whether a specified subset of $V$ is independent.
In the problem studied here, a maximum common independent set $S$ is already
given, and the task is to decide whether there is another maximum common
independent set $T\ne S$ and to output one if it exists.  We call this problem
Another Matroid Intersection.  We write $n\coloneqq |V|$,
$r\coloneqq |S|$, and $\nS\coloneqq V\setminus S$.

A standard tool for Matroid Intersection is the exchange graph associated with
a common independent set.  Its arcs represent feasible single-element
exchanges, and augmenting-path algorithms use reachability in this graph to
enlarge the current solution~\cite{Cunningham86,Schrijver}.  Although the set
$S$ in our problem is already maximum, the same graph still encodes how one
maximum solution can be exchanged for another.  More precisely, we prove in
\cref{thm:cycle-criterion} that another maximum solution exists if and only if
the exchange graph of $S$ contains a directed cycle.
This characterization gives a simple baseline: the exchange graph of $S$ has
at most $2r(n-r)=O(nr)$ arcs, and this bound is
tight in the worst case.  Each possible arc can be tested with one independence
query.  Thus one may construct the whole exchange graph and search it for a
directed cycle using $O(nr)$ queries.  This is the textbook baseline for
Another Matroid Intersection.

For many years, this baseline was not only simple but also faster than the
best-known query bounds for finding one maximum common independent set from
scratch.  Cunningham's classical algorithm is based on augmenting paths in the
exchange graph~\cite{Cunningham86}, and Chakrabarty, Lee, Sidford, Singla, and
Wong obtained an implementation using $\widetilde O(nr)$\footnote{Here and
throughout, $\widetilde O(\cdot)$ suppresses factors polylogarithmic in $n$.}
independence queries~\cite{CLS19}.
Blikstad, van den Brand, Mukhopadhyay, and Nanongkai
then broke this barrier for a range of parameters by exploring the exchange
graph implicitly~\cite{BBMN21}.  Building on their framework, Blikstad obtained
exact algorithms using $\widetilde O(nr^{3/4})$ randomized queries and
$\widetilde O(nr^{5/6})$ deterministic queries~\cite{Blikstad21}.  It would be
surprising if Another Matroid Intersection, with its
particularly simple cycle characterization, required asymptotically more
queries than finding a maximum solution itself.  This motivates seeking a
sub-$nr$ algorithm that avoids constructing the exchange graph.

\paragraph{Related work and alternative approaches.}
Another Matroid Intersection can also be solved using more general
solution-generation algorithms, including $K$-best enumeration.
After rank-$r$ truncation, Fukuda and Namiki enumerate the common bases
~\cite{FukudaNamiki95}; Kobayashi, Kurita, and Wasa give
polynomial-delay enumeration results for large maximal common independent
sets~\cite{KKW23}; and Camerini and Hamacher give algorithms for generating
the $K$ best weighted common independent sets in nonincreasing order of
weight~\cite{CameriniHamacher89}.  With unit weights, their $K=2$ case finds a
second maximum-cardinality common independent set whenever one exists.  These
more general approaches do not give the sub-$nr$ total independence-query
bound sought here.
For the decision version, a closer analogue comes from uniqueness testing for
matchings.  Gabow, Kaplan, and Tarjan test in
linear time whether a given perfect matching is unique, even in general graphs
~\cite{GabowKaplanTarjan01}; given a maximum bipartite matching, Tassa finds in
linear time all edges belonging to some maximum matching~\cite{Tassa12}.  We are
not aware of comparable oracle-efficient uniqueness results for general
matroid intersection.

Another Matroid Intersection also reduces to Weighted Matroid Intersection:
for $r\ge1$, assign weight $r+1$ to each element of $S$ and $r+2$ to each
element of $\nS$.  Then an optimum differs from $S$ if and only if another maximum common
independent set exists.  This reduction is not useful for our query goal.  The
weighted-to-unweighted framework of Huang, Kakimura, and Kamiyama incurs a
factor linear in the maximum weight $W$~\cite{HKK19}; here $W=\Theta(r)$, so
combining it with Blikstad's randomized bound gives
$\widetilde O(nr^{7/4})$ independence queries, which is already worse than
explicitly constructing the exchange graph.  Tu obtains a polylogarithmic
dependence on $W$ under the stronger rank-oracle model~\cite{Tu22}, but this
does not yield a sub-$nr$ total independence-query bound in our model.

\paragraph{Our results.}
Our main result gives randomized and deterministic query bounds without
constructing the exchange graph, as summarized in the following theorem.

\begin{theorem}\label{thm:main}
In the independence-oracle model, Another Matroid Intersection admits
\begin{enumerate}
    \item a Las Vegas algorithm using $\widetilde O(n\sqrt r)$ expected
    independence queries; and
    \item a deterministic algorithm using $\widetilde O(nr^{2/3})$
    independence queries.
\end{enumerate}
\end{theorem}

As noted above, finding another maximum common independent set amounts to
finding a directed cycle in the exchange graph.  We repeatedly remove vertices
of indegree zero, following Kahn's
topological-sorting procedure~\cite{Kahn62}.  If every vertex is removed, then
the exchange graph is acyclic and the given solution is unique by
\cref{thm:cycle-criterion}.  Otherwise the remaining graph contains a cycle,
and \cref{cor:cycle-output} shows how to recover another solution from a
chordless cycle.  The algorithmic task is thus to maintain all current
sources while the graph is revealed only through independence queries.
To do so without constructing the graph, we choose different
strategies according to a vertex's current indegree.  If the indegree is small,
we can afford to list all incoming arcs and maintain them explicitly.  If it
is large, listing the arcs is expensive, but the vertex cannot become a source
until many of its in-neighbors have been deleted, so it can safely remain
implicit for a while.  The challenge is to choose between these two strategies
without first determining every current indegree exactly.

Blikstad et al.\ introduced the \emph{heavy/light categorization} technique to
make precisely this choice~\cite[Sec.~5]{BBMN21}.  Here heavy and light refer
to high and low indegree in the current exchange graph.  The technique uses
the subset-query access developed independently by Chakrabarty et al.
~\cite[Sec.~3.2 and Lemma~11]{CLS19} and Nguy{\^e}n~\cite{Nguyen19}.
On the $\nS$ side of our deletion process, the randomized and deterministic
structures of Blikstad et al.\ apply directly: heavy vertices remain implicit
for a phase, while the incoming arcs of light vertices are listed
~\cite[Secs.~5.1 and~5.2]{BBMN21}.
The $S$ side is more difficult.  The available subset
query fixes a vertex of $\nS$ and queries a subset of $S$, but not conversely,
so transposing the preceding structure would be too expensive.  Our randomized
algorithm overcomes this asymmetry by repeatedly sampling vertices of $\nS$
and identifying all vertices of $S$ that receive an arc from at least one
sampled vertex.  The resulting detection frequencies classify the indegrees of
all vertices of $S$ collectively; after that step, source peeling proceeds as
in the heavy/light framework of Blikstad et al.

For the deterministic algorithm, we instead store in-neighbors as certificates
that an implicit vertex cannot yet become a source.  A single in-neighbor
cannot be allowed to certify too many vertices, since deleting it would then
trigger too much repair work.  Imposing a capacity controls this load, but may
make a genuinely high-degree vertex appear to have too few available
certificates.  The saturated in-neighbors causing this obstruction form a
small set, so we can enumerate only the possible arcs incident to that set.
This capacity--saturation tradeoff yields the deterministic
$\widetilde O(nr^{2/3})$ bound.

Combining binary partitioning with Uno's alternative-output
scheduling~\cite{Uno03} reduces the enumeration of maximum common independent
sets to at most two another-solution calls between consecutive outputs and
after the last output.  Hence the query delay has the same asymptotic order as
one Another Matroid Intersection call, and our faster another-solution
algorithms yield expected query delay $\widetilde O(n\sqrt r)$ or deterministic
query delay $\widetilde O(nr^{2/3})$; see \cref{sec:enumeration}.

\paragraph{Organization.}
\Cref{sec:prelim} gives the necessary matroid preliminaries.
\Cref{sec:reduction} reduces Another Matroid Intersection to source peeling in
an implicit directed bipartite graph.  \Cref{sec:randomized} presents the
randomized data structures and Las Vegas algorithm, and
\cref{sec:deterministic} gives their deterministic counterparts.
\Cref{sec:enumeration} describes the enumeration application.

%----------------------------------------------------------------

\section{Preliminaries}\label{sec:prelim}

For a set $A$ and elements $e,f$, write
\[
    A+e\coloneqq A\cup \Set{e},\qquad
    A-e\coloneqq A\setminus \Set{e},\qquad
    A-e+f\coloneqq (A-e)+f.
\]

\begin{definition}[Matroid]
A \emph{matroid} is a pair $\M=(V,\I)$ consisting of a finite set $V$, called
the \emph{ground set}, and a family $\I\subseteq 2^V$ satisfying the following
axioms:
\begin{enumerate}[label=(I\arabic*)]
    \item $\emptyset\in\I$;
    \item if $B\in\I$ and $A\subseteq B$, then $A\in\I$; and
    \item if $A,B\in\I$ and $|A|<|B|$, then there exists an element
    $e\in B\setminus A$ such that $A+e\in\I$.
\end{enumerate}
The members of $\I$ are called the \emph{independent sets} of $\M$.
\end{definition}

\begin{definition}[Matroid Intersection]
Given two matroids $\M_1=(V,\I_1)$ and $\M_2=(V,\I_2)$ on the same ground
set, a set in $\I_1\cap\I_2$ is called a \emph{common independent set}.  The
\emph{Matroid Intersection} problem asks for a common independent set of
maximum cardinality.
\end{definition}

We study the problem of finding another solution when one maximum common
independent set is already given.

\begin{definition}[Another Matroid Intersection]
Given two matroids $\M_1=(V,\I_1)$ and $\M_2=(V,\I_2)$ and a maximum common
independent set $S$, the \emph{Another Matroid Intersection} problem asks
whether there exists a maximum common independent set $T\ne S$ and, if so, to
output one such set $T$.
\end{definition}

Throughout the paper, the matroids $\M_1,\M_2$ and a maximum common independent
set $S$ are fixed.  We write
\[
    n\coloneqq |V|,\qquad r\coloneqq |S|,\qquad
    \nS\coloneqq V\setminus S.
\]

\begin{definition}[Exchange graph]
The \emph{exchange graph} of $S$ is the directed bipartite graph
$D_S=(S\cup\nS,E_S)$, where
\begin{align*}
    E_S\coloneqq{}&
    \Set{(x,y)\in S\times\nS\mid S-x+y\in\I_1}\\
    &\mathbin{\cup}
    \Set{(y,x)\in\nS\times S\mid S-x+y\in\I_2}.
\end{align*}
\end{definition}

\paragraph{Independence-oracle access.}

We work in the \emph{independence-oracle model}.  Given $i\in\Set{1,2}$ and
$A\subseteq V$, an independence query returns whether $A\in\I_i$.  The two
matroids are accessed only through such queries, and we measure the number of
queries made by an algorithm.

The following observation allows one independence query to test whether a
vertex in $\nS$ has a neighbor in an arbitrary nonempty subset of $S$; see
also~\cite[Sec.~3.2]{CLS19} and~\cite{Nguyen19}.

\begin{lemma}[Neighborhood query]\label{lem:neighborhood-query}
For every nonempty set $X\subseteq S$ and every $y\in\nS$,
\begin{align*}
    (S\setminus X)+y\in\I_1
    &\iff \text{there exists $x\in X$ such that $(x,y)\in E_S$},\\
    (S\setminus X)+y\in\I_2
    &\iff \text{there exists $x\in X$ such that $(y,x)\in E_S$}.
\end{align*}
Consequently, each of the two neighborhood conditions on the right-hand side
can be tested with one independence query.
\end{lemma}

We also use the following standard exchange property of matroids.

\begin{lemma}[Exchange matching~{\cite[Corollary~39.12a and
Theorem~39.13]{Schrijver}}]\label{lem:exchange-matching}
Let $X\subseteq S$ and $Y\subseteq\nS$ satisfy $|X|=|Y|$.  For
$i\in\Set{1,2}$, define a bipartite graph $G_i(X,Y)$ with bipartition $(X,Y)$
by
\begin{align*}
    E(G_1(X,Y))&\coloneqq \Set{xy\mid x\in X,\ y\in Y,\ (x,y)\in E_S},\\
    E(G_2(X,Y))&\coloneqq \Set{xy\mid x\in X,\ y\in Y,\ (y,x)\in E_S}.
\end{align*}

For each $i\in\Set{1,2}$, the following statements hold:
\begin{enumerate}[label=(\roman*)]
    \item if $(S\setminus X)\cup Y\in\I_i$, then $G_i(X,Y)$ has a perfect
    matching; and
    \item if $G_i(X,Y)$ has a unique perfect matching, then
    $(S\setminus X)\cup Y\in\I_i$.
\end{enumerate}
\end{lemma}

%----------------------------------------------------------------

\section{Reduction to source peeling}\label{sec:reduction}

We first characterize another maximum common independent set by a directed
cycle in the exchange graph.  We then leave the matroid setting and formulate
the remainder of the problem as source peeling in an implicitly represented
directed bipartite graph.

A directed cycle is \emph{chordless} if the subgraph induced by its vertices
contains no arcs other than the arcs of the cycle.

\begin{theorem}\label{thm:cycle-criterion}
There exists a maximum common independent set $T\ne S$ if and only if $D_S$
contains a directed cycle.
\end{theorem}

\begin{proof}
Suppose first that $T\ne S$ is a maximum common independent set.  Let
\[
    R\coloneqq S\setminus T
    \qquad\text{and}\qquad
    Y\coloneqq T\setminus S.
\]
Since $|T|=|S|$, we have $|R|=|Y|$.  By
\cref{lem:exchange-matching}(i), $G_1(R,Y)$ has a perfect matching formed by
arcs from $R$ to $Y$, and $G_2(R,Y)$ has a perfect matching formed by arcs
from $Y$ to $R$.  Their union contains a directed cycle in $D_S$.

Conversely, suppose that $D_S$ contains a directed cycle.  By taking one with
the fewest vertices, we obtain a chordless directed cycle $C$.  Let
\[
    R\coloneqq V(C)\cap S,
    \qquad
    Y\coloneqq V(C)\cap\nS,
    \qquad
    T\coloneqq (S\setminus R)\cup Y.
\]
The arcs of $C$ from $R$ to $Y$ form the unique perfect matching of
$G_1(R,Y)$, and the arcs from $Y$ to $R$ form the unique perfect matching of
$G_2(R,Y)$.  Hence \cref{lem:exchange-matching}(ii) implies
$T\in\I_1\cap\I_2$.  Moreover, $|T|=|S|$ and $T\ne S$, so $T$ is another
maximum common independent set.
\end{proof}

The converse argument gives the following useful output rule.

\begin{corollary}\label{cor:cycle-output}
For every chordless directed cycle $C$ in $D_S$, the set
$S\mathbin{\triangle}V(C)$ is a maximum common independent set different from
$S$.
\end{corollary}

It therefore suffices to find a chordless directed cycle or certify that no
directed cycle exists.  We next express this task without any further
reference to matroids.

\subsection{The directed bipartite query model}\label{sec:bipartite-query-model}

Let
\[
    G=(A\sqcup B,E),
    \qquad
    E\subseteq (A\times B)\cup(B\times A),
\]
be a directed bipartite graph with $|A|=a$ and $|B|=b$.  The sets $A$ and
$B$ are given explicitly, whereas the arc set $E$ is implicit.  We call the
vertices in $A$ \emph{rows} and those in $B$ \emph{columns}.\footnote{These
terms refer only to the two index sets of the two $0$--$1$ matrices
representing the two arc directions; no order or algebraic structure is
assumed.}

The graph is accessed through the following two Boolean oracles.  For a row
$p\in A$ and a nonempty set of columns $Q\subseteq B$, define
\begin{align*}
    \mathsf{HasIn}(p,Q)
    &\coloneqq \mathbf{1}\bigl[\text{there exists $q\in Q$ with $(q,p)\in E$}\bigr],\\
    \mathsf{HasOut}(p,Q)
    &\coloneqq \mathbf{1}\bigl[\text{there exists $q\in Q$ with $(p,q)\in E$}\bigr].
\end{align*}
Thus both oracles return a value in $\Set{0,1}$.  We count calls to these
oracles and refer to either type of call as a \emph{graph query}.  The problem
in this model is to find a chordless directed cycle in $G$ or report that $G$
is acyclic.

\begin{remark}\label{rem:exchange-query-model}
The exchange graph $D_S$ is an instance of this model with
\[
    A=\nS,\qquad B=S,\qquad a=n-r,\qquad b=r.
\]
For $p\in\nS$ and nonempty $Q\subseteq S$, \cref{lem:neighborhood-query}
implements $\mathsf{HasIn}(p,Q)$ by querying whether
$(S\setminus Q)+p\in\I_1$, and implements $\mathsf{HasOut}(p,Q)$ by querying
whether $(S\setminus Q)+p\in\I_2$.  Hence one graph query requires exactly one
independence query.  In fact, the graph model uses only the restricted class
of independence queries on sets $T$ satisfying $|T\setminus S|=1$.
\end{remark}

\paragraph{Basic primitives.}

For a row $p\in A$ and a set of columns $Q\subseteq B$,
$\mathsf{FindIn}(p,Q)$ returns a column $q\in Q$ with $(q,p)\in E$, or
$\bot$ if none exists.  The procedure $\mathsf{FindOut}(p,Q)$ is defined
analogously for an arc $(p,q)\in E$.  Both are implemented by binary search
using the corresponding Boolean oracle.

The procedures $\mathsf{ListIn}(p,Q)$ and $\mathsf{ListOut}(p,Q)$ return,
respectively,
\[
    \Set{q\in Q\mid(q,p)\in E}
    \qquad\text{and}\qquad
    \Set{q\in Q\mid(p,q)\in E}.
\]
They repeatedly find and remove one column from the candidate set.  Finally,
for $R\subseteq A$ and $Q\subseteq B$, the procedure
$\mathsf{CollectOut}(R,Q)$ returns
\[
    \Set{q\in Q\mid\text{there exists $p\in R$ with $(p,q)\in E$}
    }
\]
together with, for every returned column $q$, a witness $w(q)\in R$ such that
$(w(q),q)\in E$.  It processes the rows in $R$ while sharing one set of
columns not yet found.  Notice that we do not need a corresponding
$\mathsf{CollectIn}$ procedure.

\begin{lemma}\label{lem:primitive-costs}
The procedures $\mathsf{FindIn}$ and $\mathsf{FindOut}$ use $O(\log b)$ graph
queries.  If $k$ columns are returned, $\mathsf{ListIn}$ and
$\mathsf{ListOut}$ use $O(1+k\log b)$ graph queries.  If $k$ columns are
returned, $\mathsf{CollectOut}(R,Q)$ uses
$O(|R|+k\log b)$ graph queries.
\end{lemma}

\begin{proof}
A binary search halves the candidate set after each query, which proves the
first bound.  Each successful iteration of a listing procedure discovers and
removes one column, and there is at most one final negative query.  In
$\mathsf{CollectOut}$, each processed row contributes at most one final
negative query, while every successful search removes one column from the
shared candidate set.  Thus no returned column is charged more than once.
\end{proof}

Here and below, $O(\log b)$ abbreviates $O(\log\max\{b,2\})$.  The
procedures $\mathsf{FindIn}$ and $\mathsf{ListIn}$ are obtained from
$\mathsf{FindOut}$ and $\mathsf{ListOut}$, respectively, by replacing
$\mathsf{HasOut}$ with $\mathsf{HasIn}$.  We therefore give pseudocode only
for the outward procedures in \cref{alg:basic-primitives}.

\begin{algorithm}[H]
\caption{Out-neighbor primitives in the directed bipartite query model}
\label{alg:basic-primitives}

\Function{\FindOut{$p,Q$}}{
    \If{$Q=\emptyset$}{\Return $\bot$\;}
    \If{\HasOut{$p,Q$} $=0$}{\Return $\bot$\;}
    \While{$|Q|>1$}{
        partition $Q$ into nonempty sets $Q_1,Q_2$ with
        $\bigl||Q_1|-|Q_2|\bigr|\le 1$\;
        \eIf{\HasOut{$p,Q_1$} $=1$}{
            $Q\gets Q_1$\;
        }{
            $Q\gets Q_2$\;
        }
    }
    \Return the unique column in $Q$\;
}

\Function{\ListOut{$p,Q$}}{
    $L\gets\emptyset$\;
    $q\gets\FindOut{$p,Q$}$\;
    \While{$q\ne\bot$}{
        $L\gets L+q$\;
        $Q\gets Q-q$\;
        $q\gets\FindOut{$p,Q$}$\;
    }
    \Return $L$\;
}

\Function{\CollectOut{$R,Q$}}{
    $L\gets\emptyset$\;
    \ForEach{$p\in R$}{
        $q\gets\FindOut{$p,Q$}$\;
        \While{$q\ne\bot$}{
            $L\gets L+q$ and $w(q)\gets p$\;
            $Q\gets Q-q$\;
            $q\gets\FindOut{$p,Q$}$\;
        }
    }
    \Return $(L,w)$\;
}
\end{algorithm}

\subsection{Source peeling and cycle recovery}\label{sec:source-peeling}

We use the following standard characterization of directed acyclic graphs,
which underlies the topological-sorting algorithm of Kahn~\cite{Kahn62}.

\begin{lemma}[Source peeling]\label{lem:source-peeling}
Let $H$ be a finite directed graph.  Repeatedly delete a source and all arcs
incident to it.  This process deletes every vertex if and only if $H$ is
acyclic.  If it stops with a nonempty graph, the remaining graph has no source
and contains a directed cycle.
\end{lemma}

\begin{proof}
Deleting a source cannot destroy a directed cycle.  Hence, if $H$ contains a
cycle, all of its vertices cannot be deleted.  Conversely, every nonempty
finite directed graph with no source contains a directed cycle: starting from
any vertex and repeatedly following an incoming arc eventually revisits a
vertex.  Therefore every nonempty acyclic graph has a source, and the peeling
process deletes all of its vertices.
\end{proof}

\begin{remark}\label{rem:recover-cycle}
Suppose that source peeling stops with active sets $A'\subseteq A$ and
$B'\subseteq B$.  For every row $p\in A'$, $\mathsf{FindIn}(p,B')$ supplies
an in-neighbor.  Moreover, $\mathsf{CollectOut}(A',B')$ returns every column
in $B'$ and supplies one in-neighbor row for each of them.  Following these
chosen incoming arcs eventually finds a directed cycle.  This takes
$O((|A'|+|B'|)\log b)$ graph queries.
\end{remark}

The resulting cycle need not be chordless, but it can be shortened without
constructing the whole graph.

\begin{lemma}\label{lem:cycle-shortening}
Given a directed cycle $C$ in $G$, one can find a chordless directed cycle
whose vertex set is contained in $V(C)$ using
$O(|V(C)|\log b)$ graph queries.
\end{lemma}

\begin{proof}
For every row $p$ on the current cycle, use $\mathsf{FindIn}$ to search the
cycle columns other than the predecessor of $p$, and use $\mathsf{FindOut}$
to search the cycle columns other than the successor of $p$.  A returned arc
is a chord.  If the chord is $(u,v)$, replace the current cycle by this chord
and the directed segment of the cycle from $v$ to $u$.  The new cycle is
strictly shorter.

Each successful search shortens the cycle and hence occurs at most
$|V(C)|$ times.  A negative search need not be repeated: after a shortening,
the candidate set of every surviving row only shrinks.  There are two such
candidate sets per row.  Thus there are $O(|V(C)|)$ searches in total, each
using $O(\log b)$ queries by \cref{lem:primitive-costs}.  When all searches
are negative, the current cycle is chordless.
\end{proof}

By \cref{lem:source-peeling,rem:recover-cycle,lem:cycle-shortening}, it is
enough to perform source peeling for as long as possible.  If every vertex is
deleted, the graph is acyclic.  Otherwise, the remaining source-free graph
yields a chordless cycle and, in the exchange-graph instance, a different
maximum common independent set by \cref{cor:cycle-output}.

\paragraph{Dynamic terminology.}

The two data structures developed in the next sections share active sets
$\Aact\subseteq A$ and $\Bact\subseteq B$, initially equal to $A$ and $B$,
respectively.  The underlying sets $A,B$ and the graph $G$ never change;
only $\Aact$ and $\Bact$ shrink.  All indegrees are taken in the subgraph
induced by $\Aact\cup\Bact$.

The terms used by the two data structures are intentionally asymmetric.  The
\emph{zero-row structure} finds and deletes rows of indegree zero.  A
\emph{column deletion} may decrease row indegrees, whereas a \emph{row
deactivation} removes the monitored row itself.  The \emph{zero-column
structure} finds and deletes columns of indegree zero.  A \emph{row deletion}
may decrease column indegrees, whereas a \emph{column deactivation} removes
the monitored column itself.  Graph-theoretically, all four operations remove
one vertex from its active set.

When the structures are run together, deleting a source row is a row
deactivation for the zero-row structure and a row deletion for the
zero-column structure.  Deleting a source column is a column deletion for the
zero-row structure and a column deactivation for the zero-column structure.
Consequently, the two structures implement source peeling on the same active
graph.

%----------------------------------------------------------------

\section{Randomized algorithm}\label{sec:randomized}

Maintaining the indegree of every active vertex after every deletion would be
too expensive.  Instead, we divide the deletion sequence into phases.  At the
beginning of a phase, every vertex not already maintained explicitly is
classified.  A vertex found to have small indegree is made \emph{explicit}:
all of its active incoming arcs are listed and stored.  A vertex classified as
having large indegree remains \emph{implicit}.  The classification guarantees
that such a vertex cannot become a source before the next phase.

We develop this strategy separately for rows and columns.  The two
classifiers are asymmetric because the graph model provides queries indexed
by a single row and a subset of columns.  We assume $a,b\ge1$ throughout this
section, since a directed bipartite graph with an empty side is acyclic.

\subsection{Maintaining zero rows}\label{sec:zero-row}

The zero-row structure is the randomized heavy/light categorization of
Blikstad--van den Brand--Mukhopadhyay--Nanongkai~\cite[Sec.~5.1]{BBMN21},
specialized to our graph-query model.  It receives column deletions and row
deactivations and must emit every active row when its indegree becomes zero.
For $p\in\Aact$, let
\[
    d(p)\coloneqq
    \bigl|\Set{q\in\Bact\mid(q,p)\in E}\bigr|.
\]
This quantity changes as columns are deleted.

Fix a phase length $h\ge 1$.  A phase begins at initialization and after
every $h$ column deletions.  Row deactivations do not count toward the phase
length.  Since at most $b$ columns are deleted, the number of phases is at
most
\[
    \Phi\coloneqq 1+\left\lceil\frac{b}{h}\right\rceil.
\]
Fix also a permissible failure probability $\delta\in(0,1)$ and set
\[
    L\coloneqq
    \left\lceil 10\log\frac{2a\Phi}{\delta}\right\rceil.
\]

Every active row is in one of two states.  For an \emph{explicit} row, all
active incoming arcs are stored.  Its stored neighborhood is updated after
each column deletion, so the row is emitted exactly when that neighborhood
becomes empty.  All other active rows are \emph{implicit}.  At the beginning
of a phase, every implicit row is classified as follows.  In one experiment,
include each active column independently with probability $1/(4h)$, obtaining
a random set $Q\subseteq\Bact$, and evaluate $\mathsf{HasIn}(p,Q)$ for every
implicit row $p$.  We interpret $\mathsf{HasIn}(p,\emptyset)$ as zero without
making a query.  The $L$ experiments use mutually independent samples.  A row
is labeled \emph{heavy} if at least $3L/4$ answers are positive and
\emph{light} otherwise.  A heavy row remains implicit until the next phase.
A light row is made explicit by calling $\mathsf{ListIn}(p,\Bact)$.

The words heavy and light are classifier labels, rather than definitions in
terms of the exact degree.  The following lemma relates the labels to the
degree at the beginning of the phase.

\begin{samepage}
\begin{lemma}\label{lem:row-classification}
For an implicit row $p$ with degree $d(p)$ at the beginning of a phase,
\begin{align*}
    \Pr[p\text{ is labeled heavy}\mid d(p)<2h]&\le e^{-L/8},\\
    \Pr[p\text{ is labeled light}\mid d(p)\ge 16h]&\le e^{-L/10}.
\end{align*}
\end{lemma}
\end{samepage}

\begin{proof}
The concentration calculation is given in
\cref{app:row-classification-calculation}.
\end{proof}

\begin{theorem}\label{thm:zero-row}
For every adaptive sequence of column deletions and row deactivations, the
zero-row structure emits exactly all active rows of indegree zero with
probability at least $1-\delta$.  It uses
\[
    O\left(
        a\left(1+\frac{b}{h}\right)
        \log\frac{a(1+b/h)}{\delta}
        +ah\log b
    \right)
\]
graph queries in total.
\end{theorem}

\begin{proof}
Condition on the complete execution history before a phase begins.  The
active graph is then fixed, whereas the samples used in the new phase are
fresh.  By \cref{lem:row-classification}, the conditional probability that a
particular row--phase classification violates either of the implications
\[
    \text{heavy}\Longrightarrow d(p)\ge2h,
    \qquad
    \text{light}\Longrightarrow d(p)<16h
\]
is at most $e^{-L/10}\le\delta/(2a\Phi)$.  There are at most $a\Phi$
classifications.  A union bound, applied successively under the above
conditioning, therefore gives both implications simultaneously with
probability at least $1-\delta$.  Classifications of different rows may be
correlated because they use the same sampled sets; the union bound does not
require them to be independent.

On this good event, a heavy row starts a phase with degree at least $2h$.
At most $h$ columns are deleted before the next classification, so the row
retains positive degree throughout the phase.  A light row has fewer than
$16h$ active in-neighbors when it becomes explicit, and its stored
neighborhood is subsequently maintained exactly.  Hence no source is missed
and no row is emitted before becoming a source.

The experiments use $O(a\Phi L)$ queries.  Each row becomes explicit at most
once, and \cref{lem:primitive-costs} bounds all calls to
$\mathsf{ListIn}$ by $O(a+ah\log b)$ queries.  Substituting
$\Phi=O(1+b/h)$ and
$L=O(\log(a(1+b/h)/\delta))$ gives the stated query bound.
\end{proof}

\subsection{Maintaining zero columns}\label{sec:zero-column}

We next develop the additional ingredient needed to categorize columns.  The
zero-column structure receives row deletions and column deactivations and must
emit every active column when its indegree becomes zero.  For
$q\in\Bact$, write
\[
    d(q)\coloneqq
    \bigl|\Set{p\in\Aact\mid(p,q)\in E}\bigr|.
\]
Unlike the row case, our oracle cannot directly test whether one column has an
in-neighbor in a specified subset of rows.  We instead classify all implicit
columns collectively by sampling rows and calling $\mathsf{CollectOut}$.

Fix a phase length $h\ge1$.  Here $h$ is a parameter of the zero-column
structure and may be chosen separately from the phase length used by the
zero-row structure.  A phase begins at initialization and after every $h$ row
deletions; column deactivations do not count.  The number of phases is at most
\[
    \Psi\coloneqq1+\left\lceil\frac{a}{h}\right\rceil.
\]
At the beginning of a phase, let $C\subseteq\Bact$ be the set of implicit
columns and put $s\coloneqq|\Aact|$.  If $s<4h$, every column is labeled
light.  Otherwise, set
\[
    k\coloneqq\left\lfloor\frac{s}{4h}\right\rfloor,
    \qquad
    L\coloneqq\left\lceil50\log\frac{2b\Psi}{\delta}\right\rceil,
\]
and perform $L$ independent collective experiments.  In one experiment, it
samples a uniformly random $k$-element set $R\subseteq\Aact$ and calls
$\mathsf{CollectOut}(R,C)$.  A column is \emph{detected} in that experiment
if it belongs to the returned set $D$.  The counter $z(q)$ records the number
of experiments in which $q$ is detected; the witnesses returned by
$\mathsf{CollectOut}$ are not used.  A column is labeled heavy if
$z(q)\ge3L/4$ and light otherwise.  The complete procedure is given in
\cref{alg:column-categorization}; it returns a partition
$(C_{\mathrm{heavy}},C_{\mathrm{light}})$ of $C$.

\begin{algorithm}[H]
\caption{Randomized categorization of implicit columns at a phase start}
\label{alg:column-categorization}
\Procedure{\ClassifyColumns{$C,h,\delta$}}{
    $L\gets\left\lceil50\log(2b\Psi/\delta)\right\rceil$\;
    \If{$C=\emptyset$}{
        \Return $(\emptyset,\emptyset)$\;
    }
    \If{$|\Aact|<4h$}{
        \Return $(\emptyset,C)$\;
    }
    $k\gets\left\lfloor |\Aact|/(4h)\right\rfloor$\;
    initialize $z(q)\gets0$ for every $q\in C$\;
    \For{$\ell\gets1$ \KwTo $L$}{
        choose a uniformly random $k$-element set $R\subseteq\Aact$\;
        $(D,w)\gets\CollectOut{$R,C$}$\;
        \ForEach{$q\in D$}{
            $z(q)\gets z(q)+1$\;
        }
    }
    $C_{\mathrm{heavy}}\gets
        \Set{q\in C\mid z(q)\ge3L/4}$\;
    $C_{\mathrm{light}}\gets C\setminus C_{\mathrm{heavy}}$\;
    \Return $(C_{\mathrm{heavy}},C_{\mathrm{light}})$\;
}
\end{algorithm}

Every heavy column remains implicit.  All light columns are made explicit
simultaneously: call $\mathsf{ListOut}(p,C_{\mathrm{light}})$ for every
$p\in\Aact$.  These calls give the exact active in-neighborhood of every
column in $C_{\mathrm{light}}$.  The stored neighborhoods are then updated
after each row deletion, and an explicit column is emitted exactly when its
neighborhood becomes empty.

\begin{lemma}
\label{lem:column-classification}
For an implicit column $q$ with degree $d(q)$ at the beginning of a phase,
\begin{align*}
    \Pr[q\text{ is labeled heavy}\mid d(q)<2h]&\le e^{-L/8},\\
    \Pr[q\text{ is labeled light}\mid d(q)\ge16h]&\le e^{-L/50}.
\end{align*}
If $|\Aact|<4h$, the relevant implications hold deterministically.
\end{lemma}

\begin{proof}
The concentration calculation is given in
\cref{app:column-classification-calculation}.
\end{proof}

\begin{theorem}\label{thm:zero-column}
For every adaptive sequence of row deletions and column deactivations, the
zero-column structure emits exactly all active columns of indegree zero with
probability at least $1-\delta$.  Its total number of graph queries is
\[
    O\left(
        \left(
            \frac{a^2}{h^2}+\frac{a}{h}
            +\left(b+\frac{ab}{h}\right)\log b
        \right)
        \log\frac{b(1+a/h)}{\delta}
        +a+\frac{a^2}{h}+bh\log b
    \right).
\]
\end{theorem}

\begin{proof}
As in the row case, condition on the complete history before each phase.
Fresh uniform samples give the probabilities in
\cref{lem:column-classification}, even when the deletion sequence is
adaptive.  There are at most $b\Psi$ column--phase classifications, and our
choice of $L$ makes the conditional failure probability of each at most
$\delta/(2b\Psi)$.  A union bound therefore implies, with probability at
least $1-\delta$, that every heavy column starts its phase with degree at
least $2h$ and every light column has degree below $16h$.  The detection
events of different columns need not be independent.

On this good event, an implicit heavy column retains positive degree during
the next at most $h$ row deletions.  A light column becomes explicit with
fewer than $16h$ active in-neighbors and is maintained exactly thereafter.
Thus the structure emits precisely the zero-degree columns.

Let $s_j$ be the number of active rows at the beginning of phase $j$.  Since
$s_j\le\max\Set{a-(j-1)h,0}$, summing this arithmetic progression gives
$\sum_j s_j=O(a+a^2/h)$ and
$\sum_j s_j/h=O(a/h+a^2/h^2)$.  By
\cref{lem:primitive-costs}, the $L$ collective experiments in phase $j$ use
$O(L(s_j/h+b\log b))$ queries.  Hence all classification experiments use
\[
    O\left(
        L\left(
            \frac{a^2}{h^2}+\frac{a}{h}
            +\left(b+\frac{ab}{h}\right)\log b
        \right)
    \right)
\]
queries.
Making newly light columns explicit contributes one final negative query per
active row and phase, totaling $O(a+a^2/h)$.  Each column becomes explicit
only once and then has fewer than $16h$ in-neighbors, so all successful
searches in these calls cost $O(bh\log b)$.  Combining these costs and using
$L=O(\log(b(1+a/h)/\delta))$ proves the theorem.
\end{proof}

\subsection{Source peeling and the randomized bound}
\label{sec:randomized-source-peeling}

We now run the two structures together on the directed bipartite graph $G$.
Whenever either structure emits a source, we delete that vertex and notify
both structures as described at the end of \cref{sec:source-peeling}.  Thus
the common active sets follow exactly the source-peeling process.

For the zero-row structure choose
\[
    h_1\coloneqq\lceil\sqrt b\rceil,
\]
and for the zero-column structure choose
\[
    h_2\coloneqq
    \min\Set{a,\max\Set{1,
        \left\lceil\sqrt{\frac{a(a+b)}{b}}\right\rceil}}.
\]

\begin{theorem}\label{thm:randomized-graph}
For every constant $c>0$, there is a randomized algorithm that, with
probability at least $1-(a+b)^{-c}$, finds a chordless directed cycle in $G$
or correctly reports that $G$ is acyclic, using
\[
    \widetilde O\left(
        a+b+\sqrt{ab(a+b)}+a\sqrt b
    \right)
\]
graph queries.  Moreover, there is a Las Vegas algorithm with the same
expected query bound.
\end{theorem}

\begin{proof}
Set $\delta=\frac12(a+b)^{-c}$ for each of the two structures.  By
\cref{thm:zero-row,thm:zero-column} and a union bound, their joint good event
has probability at least $1-(a+b)^{-c}$.  No independence between the two
structures is required.  On this event, source peeling either deletes every
vertex, certifying acyclicity, or stops at a nonempty source-free graph.  In
the latter case, \cref{rem:recover-cycle,lem:cycle-shortening} produces a
chordless directed cycle using $\widetilde O(a+b)$ further queries.

The zero-row structure costs $\widetilde O(a\sqrt b+a)$.  The zero-column
structure costs
\[
    \widetilde O\left(
        \frac{a(a+b)}{h_2}+bh_2+a+b
    \right)
    =\widetilde O\left(a+b+\sqrt{ab(a+b)}\right).
\]
Indeed, this follows by balancing the first two terms.  If the square-root
choice exceeds $a$, then $a+b>ab$, and taking $h_2=a$ makes the expression
$O(a+b)$.  This proves the high-probability bound.

For completeness, we convert the procedure to a Las Vegas algorithm.  The
randomized structures have one-sided error: a vertex is emitted only after
its explicitly stored neighborhood has become empty.  Run independent trials,
each with a constant failure parameter and a query cap equal to a sufficiently
large constant times the displayed bound.  If a trial stops with a nonempty
active graph, use $\mathsf{FindIn}$ on every active row and
$\mathsf{CollectOut}$ on all active rows to verify that no source remains.
If the cap is exceeded or the verification finds a source, restart.  Otherwise
the trial either has deleted every vertex, which certifies acyclicity, or the
verified source-free graph yields a chordless cycle.  A trial succeeds with
constant probability, so the expected number of trials is constant.
\end{proof}

Applying \cref{thm:randomized-graph} to the exchange graph gives the desired
oracle bound.

\begin{corollary}\label{cor:randomized-matroid}
Another Matroid Intersection admits a Las Vegas algorithm using
\[
    \widetilde O(n\sqrt r)
\]
expected independence queries.
\end{corollary}

\begin{proof}
If $r=0$ or $n-r=0$, the given solution is unique.  Otherwise substitute
$a=n-r$ and $b=r$ in \cref{thm:randomized-graph}.  By
\cref{rem:exchange-query-model}, every graph query is one independence query.
Furthermore,
\[
    n+\sqrt{(n-r)rn}+(n-r)\sqrt r=O(n\sqrt r).
\]
A chordless cycle yields another maximum common independent set by
\cref{cor:cycle-output}.
\end{proof}

%----------------------------------------------------------------

\section{Deterministic algorithm}\label{sec:deterministic}

We now replace the randomized classifications by deterministic witness
structures.  For every active vertex not yet known to be a source, we store
explicit evidence of positive indegree.  When an evidence vertex is deleted,
the structure either finds replacement evidence or discovers that the
monitored vertex has become a source.

We write $\mathsf{Emit}(v)$ for inserting a newly discovered source $v$ into
a queue maintained by the source-peeling driver.  The insertion itself makes
no graph query and does not interrupt the current data-structure operation.
After that operation returns, the driver repeatedly extracts a queued vertex
that is still active, removes it from the common active graph, and invokes the
corresponding deactivation and deletion operations of the two structures.
Thus emissions produced by those operations are processed only after the
operation that produced them has returned.  Duplicate emissions may simply be
discarded.  At initialization, both structures are initialized before the
driver starts processing the accumulated queue.  As before, all explicitly
stored sets, counters, and reverse lists are available without oracle queries.

\subsection{Deterministic zero rows by minimum-load witnesses}
\label{sec:det-zero-row}

The zero-row structure in this subsection is the deterministic
categorization of Blikstad--van den Brand--Mukhopadhyay--Nanongkai
~\cite[Sec.~5.2]{BBMN21}, specialized to zero-row maintenance.  It receives
column deletions and row deactivations and emits rows when their indegree
becomes zero.

For every active row $p$ not yet emitted, store either an active in-neighbor
$w(p)\in\Bact$, called its \emph{witness}, or the value $\bot$.  For each
active column $q$, maintain the reverse list
\[
    \mathsf{Users}(q)
    \coloneqq\Set{p\in\Aact\mid w(p)=q}
\]
and its \emph{load} $\lambda(q)\coloneqq|\mathsf{Users}(q)|$.

Whenever a row needs a witness, order the active columns by nondecreasing
load, breaking ties by an arbitrary fixed order.  Binary search on prefixes
with $\mathsf{HasIn}$ returns an in-neighbor of minimum load or certifies that
none exists.  We denote this procedure by
$\mathsf{MinLoadNeighbor}$.  The complete structure is shown in
\cref{alg:det-zero-row}.  In the comments attached to $\mathsf{Emit}$, ``report''
means handing the vertex to the source-peeling driver described above.

\begin{algorithm}[H]
\caption{Deterministic zero-row maintenance by minimum-load witnesses}
\label{alg:det-zero-row}

\Function{\MinLoadNeighbor{$p$}}{
    let $q_1,\ldots,q_t$ be the columns in $\Bact$, ordered by
    $(\lambda(q),q)$\;
    \If{$t=0$}{\Return $\bot$\;}
    \If{\HasIn{$p,\Set{q_1,\ldots,q_t}$} $=0$}{
        \Return $\bot$\;
    }
    binary search for the least $j$ such that
    $\HasIn(p,\Set{q_1,\ldots,q_j})=1$\;
    \Return $q_j$\;
}

\Procedure{\AssignWitness{$p$}}{
    $q\gets\MinLoadNeighbor{$p$}$\;
    \eIf{$q=\bot$}{
        $\Emit(p)$\Comment*[r]{report $p$ as a zero row}
    }{
        $w(p)\gets q$\;
        $\mathsf{Users}(q)\gets\mathsf{Users}(q)+p$\;
        $\lambda(q)\gets\lambda(q)+1$\;
    }
}

\Procedure{\InitializeZeroRows{}}{
    initialize $w(p)\gets\bot$ for $p\in\Aact$ and
    $\mathsf{Users}(q)\gets\emptyset$, $\lambda(q)\gets0$ for $q\in\Bact$\;
    \ForEach{$p\in\Aact$}{\AssignWitness{$p$}\;}
}

\Procedure{\DeleteColumn{$q$}}{
    $U\gets\mathsf{Users}(q)$ and $\Bact\gets\Bact-q$\;
    \ForEach{$p\in U$}{
        $w(p)\gets\bot$\;
        $\AssignWitness(p)$\;
    }
    $\mathsf{Users}(q)\gets\emptyset$ and $\lambda(q)\gets0$\;
}

\Procedure{\DeactivateRow{$p$}}{
    $\Aact\gets\Aact-p$\;
    \If{$w(p)\ne\bot$}{
        $q\gets w(p)$\;
        $\mathsf{Users}(q)\gets\mathsf{Users}(q)-p$\;
        $\lambda(q)\gets\lambda(q)-1$ and $w(p)\gets\bot$\;
    }
}
\end{algorithm}

\begin{theorem}
\label{thm:det-zero-row}
The minimum-load witness structure emits exactly all active rows of indegree
zero under every sequence of column deletions and row deactivations.  It uses
\[
    O\bigl((a+b+a\sqrt b)\log b\bigr)
\]
graph queries in total.
\end{theorem}

\begin{proof}
Whenever $\mathsf{AssignWitness}(p)$ is called, $w(p)=\bot$.  A successful
search stores a genuine active in-neighbor, whereas a failed search certifies
that no active in-neighbor exists and emits $p$.  Thus the structure is
correct.  Since active sets only shrink, an emitted row never becomes
positive again.

It remains to bound the number of successful witness assignments.  Assume
$a,b\ge1$ and put $\theta\coloneqq a/\sqrt b$.  Call an assignment to $q$
\emph{low-load} if $\lambda(q)\le\theta$ immediately before the assignment,
and \emph{high-load} otherwise.  Before a fixed column is deleted, it receives
at most $\lfloor\theta\rfloor+1$ low-load assignments, except for assignments
that compensate for earlier load decrements caused by row deactivations.
Every column is deleted at most once, and every row deactivation causes at
most one such decrement.  Hence the number of low-load assignments is at most
$b(\theta+1)+a=O(a\sqrt b+a+b)$.

Consider the first high-load assignment made to a fixed row $p$.  Since the
chosen column has minimum load among all active in-neighbors of $p$, every
such in-neighbor then has load greater than $\theta$.  At all times,
\[
    \sum_{q\in\Bact}\lambda(q)\le a,
\]
so $p$ has fewer than $a/\theta=\sqrt b$ active in-neighbors at that moment.
Every subsequent witness assigned to $p$ is a different member of this
shrinking set, because a witness is replaced only after its column is
deleted.  Thus there are $O(a\sqrt b)$ high-load assignments in total.

Each successful assignment and the final failed search of each emitted row
uses $O(\log b)$ queries.  The claimed bound follows.
\end{proof}

\subsection{Deterministic zero columns by capacity saturation}
\label{sec:det-zero-column}

We next give the new deterministic ingredient for zero-column maintenance;
the deterministic categorization of~\cite[Sec.~5.2]{BBMN21} does not supply this
transposed structure under the asymmetric row-subset query interface.  It
receives row deletions and column deactivations and emits columns when their
indegree becomes zero.  For $q\in\Bact$, write
\[
    N(q)\coloneqq\Set{p\in\Aact\mid(p,q)\in E}
\]
for its current active in-neighborhood.

Fix integers $h,\tau$ with $1\le h\le a$ and $\tau\ge1$.  The phase length
$h$ may be chosen independently of all phase lengths used earlier.  It is
also the number of distinct witness rows required by a regular column at a
phase start.  The \emph{capacity} $\tau$ is the maximum number of regular
columns for which one row may simultaneously serve as a witness.

The capacity controls cumulative witness turnover.  Without it, one deleted
row could destroy witnesses of all $b$ columns, and repeated deletions could
cause $\Theta(ab)$ assignments.  A capacity of $\tau$ limits the number of
incidences destroyed by one row deletion to $\tau$.  Its cost is that a
column may fail to acquire $h$ witnesses because adjacent rows are already
saturated.  We then enumerate those saturated rows collectively.  The terms
$a\tau$ and $b^2h/\tau$ in the analysis below quantify this tradeoff.

An active column is in exactly one of two states:
\begin{itemize}
    \item A \emph{regular column} $q$ stores a witness set
    $W(q)\subseteq N(q)$ of distinct active rows.  Immediately after a
    repair, $|W(q)|=h$.
    \item An \emph{explicit column} $q$ stores its entire active
    neighborhood $\mathsf{Adj}(q)=N(q)$.  It never becomes regular again.
\end{itemize}
A pair $(p,q)$ with $p\in W(q)$ is a \emph{witness incidence}.  The
\emph{row load}
\[
    \lambda(p)\coloneqq
    \bigl|\Set{q\in\Bact\mid q\text{ is regular and }p\in W(q)}\bigr|
\]
is always at most $\tau$.  A regular column is \emph{deficient} during a
repair if it has fewer than $h$ witnesses.  Initialization performs a repair,
and another repair is triggered after every $h$ row deletions.  Column
deactivations do not count toward the phase length.  We use $c$ for the number
of row deletions since the most recent repair, so $0\le c<h$ between repairs.

\Cref{alg:capacity-repair} gives the repair procedure.  The set
$R_{\mathrm{sat}}$ is defined before any load of a newly explicit column is
released; this timing is essential for the certificate proved below.

\begin{algorithm}[H]
\caption{$\mathsf{RepairColumns}(h,\tau)$: capacity-saturation repair}
\label{alg:capacity-repair}
\Procedure{\RepairColumns{$h,\tau$}}{
    $C_{\mathrm{def}}\gets
        \Set{q\in\Bact\mid q\text{ is regular and }|W(q)|<h}$\;
    \ForEach{$p\in\Aact$ in a fixed order}{
        \While{$\lambda(p)<\tau$}{
            $C_p\gets\Set{q\in C_{\mathrm{def}}\mid p\notin W(q)}$\;
            $q\gets\FindOut{$p,C_p$}$\;
            \If{$q=\bot$}{\Break\;}
            $W(q)\gets W(q)+p$ and $\lambda(p)\gets\lambda(p)+1$\;
            \If{$|W(q)|=h$}{
                $C_{\mathrm{def}}\gets C_{\mathrm{def}}-q$\;
            }
        }
    }
    $R_{\mathrm{sat}}\gets\Set{p\in\Aact\mid\lambda(p)=\tau}$
        \Comment*[r]{define before releasing any load}
    \ForEach{$q\in C_{\mathrm{def}}$}{
        $\mathsf{Adj}(q)\gets W(q)$\;
    }
    \ForEach{$p\in R_{\mathrm{sat}}$}{
        $H\gets\ListOut{$p,C_{\mathrm{def}}$}$\;
        \ForEach{$q\in H$}{
            $\mathsf{Adj}(q)\gets\mathsf{Adj}(q)+p$\;
        }
    }
    \ForEach{$q\in C_{\mathrm{def}}$}{
        mark $q$ explicit\;
        \ForEach{$p\in W(q)$}{
            $\lambda(p)\gets\lambda(p)-1$\;
        }
        $W(q)\gets\emptyset$\;
        \If{$\mathsf{Adj}(q)=\emptyset$}{
            $\Emit(q)$\Comment*[r]{report $q$ as a zero column}
        }
    }
    $c\gets0$\Comment*[r]{$c$ counts row deletions since this repair}
}
\end{algorithm}

Initially every active column is regular with an empty witness set.  The
update procedures surrounding repairs are given in
\cref{alg:det-zero-column-updates}.  Deleting a row removes all of its stored
incidences and deletes it from every explicit neighborhood.  Deactivating a
regular column releases its witness incidences.  Explicit neighborhoods are
maintained exactly.

\begin{algorithm}[H]
\caption{Initialization and updates for deterministic zero-column maintenance}
\label{alg:det-zero-column-updates}
\Procedure{\InitializeZeroColumns{$h,\tau$}}{
    initialize $\lambda(p)\gets0$ for every $p\in\Aact$\;
    mark every $q\in\Bact$ regular and set $W(q)\gets\emptyset$\;
    $\RepairColumns(h,\tau)$\;
}

\Procedure{\DeleteRow{$p$}}{
    $\Aact\gets\Aact-p$\;
    \ForEach{regular $q\in\Bact$ with $p\in W(q)$}{
        $W(q)\gets W(q)-p$\;
    }
    $\lambda(p)\gets0$\;
    $H\gets\Set{q\in\Bact\mid q\text{ is explicit and }
        p\in\mathsf{Adj}(q)}$\;
    \ForEach{$q\in H$}{
        $\mathsf{Adj}(q)\gets\mathsf{Adj}(q)-p$\;
        \If{$\mathsf{Adj}(q)=\emptyset$}{
            $\Emit(q)$\Comment*[r]{report $q$ as a zero column}
        }
    }
    $c\gets c+1$\;
    \If{$c=h$}{$\RepairColumns(h,\tau)$\;}
}

\Procedure{\DeactivateColumn{$q$}}{
    $\Bact\gets\Bact-q$\;
    \If{$q$ is regular}{
        \ForEach{$p\in W(q)$}{
            $\lambda(p)\gets\lambda(p)-1$\;
        }
        $W(q)\gets\emptyset$\;
    }
    discard $\mathsf{Adj}(q)$ if $q$ is explicit\;
}
\end{algorithm}

\begin{lemma}\label{lem:saturation-certificate}
At the end of the row scan in \cref{alg:capacity-repair}, before any witness
load of a newly explicit column is released,
\[
    N(q)\subseteq W(q)\cup R_{\mathrm{sat}}
    \quad(q\in C_{\mathrm{def}}),
    \qquad
    |R_{\mathrm{sat}}|\le\frac{bh}{\tau}.
\]
Consequently, the enumeration over
$R_{\mathrm{sat}}\times C_{\mathrm{def}}$ computes
$\mathsf{Adj}(q)=N(q)$ exactly for every $q\in C_{\mathrm{def}}$.
\end{lemma}

\begin{proof}
Fix $q\in C_{\mathrm{def}}$ at the end of the scan and
$p\in N(q)\setminus W(q)$.  Witness sets only grow during the scan, and a
column leaves $C_{\mathrm{def}}$ as soon as it obtains $h$ witnesses.  Hence
$q$ remained deficient throughout the scan.  When $p$ was processed,
$q\in C_p$.  Since $(p,q)\in E$, the loop could not have stopped with
$\mathsf{FindOut}(p,C_p)=\bot$ while $p$ was unsaturated.  It therefore ended
with $\lambda(p)=\tau$, so $p\in R_{\mathrm{sat}}$.  This proves the
containment.

Immediately before loads are released, every regular column has at most $h$
witnesses.  Double-counting witness incidences gives
\[
    \tau|R_{\mathrm{sat}}|
    \le\sum_{p\in\Aact}\lambda(p)
    =\sum_{\substack{q\in\Bact\\q\text{ regular}}}|W(q)|
    \le bh.
\]
The stored witnesses and all neighbors found by listing the saturated rows
therefore contain every member of $N(q)$, and every stored pair is a genuine
arc.  Thus $\mathsf{Adj}(q)=N(q)$ exactly.
\end{proof}

\begin{lemma}\label{lem:capacity-phase-safety}
Immediately after a repair, every active regular column has exactly $h$
distinct active witnesses.  After $c<h$ row deletions in the current phase,
it has at least $h-c>0$ surviving witnesses.  Hence a regular column cannot
become a source without being made explicit and checked by the repair
triggered by the $h$-th deletion.
\end{lemma}

\begin{proof}
After the row scan, every regular column outside $C_{\mathrm{def}}$ has
exactly $h$ witnesses, while every column remaining in $C_{\mathrm{def}}$
becomes explicit.  One row deletion removes at most one distinct witness from
a fixed column.  The lower bound follows.  The $h$-th deletion triggers a
repair before the update returns, and \cref{lem:saturation-certificate} makes
every still-deficient column explicit with its exact neighborhood.
\end{proof}

\begin{theorem}
\label{thm:det-zero-column-general}
For $1\le h\le a$ and $\tau\ge1$, the capacity-saturation structure emits
exactly all active columns of indegree zero under every sequence of row
deletions and column deactivations.  Its total number of graph queries is
\[
    O\left(
        \left(
            a+\frac{a^2}{h}+a\tau+bh+\frac{ab}{\tau}
            +\frac{b^2h}{\tau}
        \right)\log b
    \right).
\]
\end{theorem}

\begin{proof}
Correctness follows from \cref{lem:capacity-phase-safety} for regular columns
and exact maintenance of $\mathsf{Adj}(q)$ for explicit columns.  If the
execution stops after fewer than $h$ deletions in the current phase, every
regular column still has a witness.  Explicit columns are emitted exactly
when their stored neighborhoods become empty.

There are at most $1+\lfloor a/h\rfloor$ repairs.  If $a_j$ is the number of
active rows in repair $j$, summing over the deletion phases gives
$\sum_j a_j=O(a+a^2/h)$.  Each scanned row contributes at most one terminal
negative query, and each successful witness assignment costs $O(\log b)$
queries.

Let $J$ be the total number of successful witness assignments.  An incidence
can disappear either when its row is deleted or when its regular column is
deactivated or becomes explicit.  A deleted row has load at most $\tau$, so
row deletions remove at most $a\tau$ incidences in total.  Each column leaves
the regular state at most once and then releases at most $h$ incidences; at
most $bh$ incidences remain at the end.  Conservation of incidences therefore
gives $J\le a\tau+2bh$.

In one repair, \cref{lem:saturation-certificate} gives
$|R_{\mathrm{sat}}|\le bh/\tau$.  The terminal negative queries made while
listing saturated rows over all repairs are consequently
\[
    O\left(\left(1+\frac{a}{h}\right)\frac{bh}{\tau}\right)
    =O\left(\frac{bh}{\tau}+\frac{ab}{\tau}\right).
\]
Every column becomes explicit at most once.  At that transition, the same
certificate gives
\[
    |N(q)|\le |W(q)|+|R_{\mathrm{sat}}|
    \le h+\frac{bh}{\tau}.
\]
Charging every positive discovery to the column becoming explicit yields at
most $bh+b^2h/\tau$ such discoveries over the entire execution.  Combining
the row scans, witness assignments, terminal queries, and positive
discoveries, and upper-bounding every term by its product with $O(\log b)$,
proves the theorem.
\end{proof}

\begin{corollary}
\label{cor:det-zero-column}
For $a,b\ge1$, deterministic zero-column maintenance uses
\[
    O\left(
        \left(
            a+b+
            \begin{cases}
                b\sqrt a, & a<b^{2/3},\\
                ab^{2/3}, & a\ge b^{2/3}
            \end{cases}
        \right)\log b
    \right)
\]
graph queries.
Equivalently, the bound is
\[
    O\left(
        \left(a+b+\max\Set{b\sqrt a,ab^{2/3}}\right)\log b
    \right).
\]
\end{corollary}

\begin{proof}
For fixed $h$, choose $\tau\coloneqq\lceil b\sqrt{h/a}\rceil$.  The bound
in \cref{thm:det-zero-column-general} becomes
\[
    O\left(
        \left(a+\frac{a^2}{h}+b\sqrt{ah}+bh\right)\log b
    \right).
\]
Indeed, the terms involving $\tau$ satisfy
$a\tau=O(b\sqrt{ah}+a)$,
$b^2h/\tau=O(b\sqrt{ah})$, and
$ab/\tau=O(a^2/h)$.

If $a<b^{2/3}$, take $h=1$.  Then $a^2<b\sqrt a$, giving
$O((a+b\sqrt a)\log b)$.  If $a\ge b^{2/3}$, take
$h=\lceil a/b^{2/3}\rceil$.  Here $1\le h\le a$, and every nonconstant term
in the preceding bound is $O(ab^{2/3})$.  Adding the harmless $b$ term gives
the displayed uniform statement.
\end{proof}

\subsection{Putting the deterministic structures together}
\label{sec:deterministic-combination}

Run deterministic zero-row maintenance and deterministic zero-column
maintenance on the same directed bipartite graph $G$, using the shared active
sets and source queue described in \cref{sec:source-peeling}.  The two
structures are exact, so the queue performs deterministic source peeling.

\begin{theorem}\label{thm:deterministic-graph}
For $a,b\ge1$, one can find a chordless directed cycle in $G$ or correctly
report that $G$ is acyclic using
\[
    O\left(
        \left(a+b+\max\Set{b\sqrt a,ab^{2/3}}\right)\log b
    \right)
\]
graph queries deterministically.
\end{theorem}

\begin{proof}
By \cref{thm:det-zero-row,thm:det-zero-column-general}, every emitted vertex
is a current source, and every current source is eventually emitted.  Thus the
shared queue implements the source-peeling process.  If every vertex is
removed, \cref{lem:source-peeling} certifies acyclicity.  Otherwise the
remaining graph is nonempty and source-free, and
\cref{rem:recover-cycle,lem:cycle-shortening} returns a chordless directed
cycle.

The zero-row structure uses $O((a+b+a\sqrt b)\log b)$ queries.  The optimized
zero-column bound is given by \cref{cor:det-zero-column}.  If
$a<b^{2/3}$, then $a\le b$ and $a\sqrt b\le b\sqrt a$.  If
$a\ge b^{2/3}$, then $a\sqrt b\le ab^{2/3}$.  Thus the zero-row cost and the
$O((a+b)\log b)$ cycle-recovery cost are absorbed by the displayed bound.
\end{proof}

\begin{corollary}
\label{cor:deterministic-matroid}
Another Matroid Intersection can be solved deterministically using
\[
    O\left(
        \left(
            n+\max\Set{r\sqrt{n-r},(n-r)r^{2/3}}
        \right)\log r
    \right)
    =\widetilde O(nr^{2/3})
\]
independence queries.
\end{corollary}

\begin{proof}
If $r=0$ or $n-r=0$, the given solution is unique.  Otherwise apply
\cref{thm:deterministic-graph} with $a=n-r$ and $b=r$.  By
\cref{rem:exchange-query-model}, every graph query is one independence query,
and a returned chordless cycle gives another maximum common independent set
by \cref{cor:cycle-output}.

For the uniform bound, if $n-r<r^{2/3}$, then
$r\sqrt{n-r}<r^{4/3}\le nr^{2/3}$.  In the other case,
$(n-r)r^{2/3}\le nr^{2/3}$.
\end{proof}

\begin{proof}[Proof of \cref{thm:main}]
If $r=0$ or $n-r=0$, the given solution is unique.  Otherwise the randomized
claim follows from \cref{cor:randomized-matroid}, and the deterministic claim
follows from \cref{cor:deterministic-matroid}.
\end{proof}

%----------------------------------------------------------------

\section{Enumerating maximum common independent sets}
\label{sec:enumeration}

Kobayashi, Kurita, and Wasa~\cite{KKW23} gave a polynomial-delay,
polynomial-space algorithm for enumerating large maximal common independent
sets; in particular, their framework enumerates all maximum common independent
sets by recursively imposing inclusion and exclusion constraints.  We use the
same binary partition.  Since every subproblem reached below already carries a
maximum solution, however, it can be tested by asking for another solution
rather than solving a fresh Matroid Intersection instance.  Together with the
alternative-output scheduling of Uno~\cite{Uno03}, this ensures that
backtracking itself makes no oracle query and improves the independence-query
delay for maximum solutions.  The reduction from another-solution queries to
enumeration and the alternative-output scheduling are standard; the new
ingredient used here is the faster another-solution algorithm developed in the
preceding sections.

Let $\mathcal A(n,r)$ be the worst-case number of independence queries used by
an algorithm for Another Matroid Intersection on an instance with at most $n$
elements and maximum solution size at most $r$, including a call that reports
that no other solution exists.  The \emph{query delay} of a deterministic
enumeration algorithm is the maximum number of independence queries made before
the first output, between two consecutive outputs, or between the last output
and termination.  For a Las Vegas algorithm, \emph{expected query delay} means
the maximum, over the corresponding output histories, of the conditional
expected number of queries in such an interval.

\begin{theorem}
\label{thm:enumeration}
Suppose that one maximum common independent set is given and that Another
Matroid Intersection can be solved in $\mathcal A(n,r)$ independence queries
and polynomial space.  Then all maximum common independent sets can be
enumerated without repetition, in polynomial space and with query delay
$O(\mathcal A(n,r))$.  More precisely, at most two another-solution calls are
made between consecutive outputs and between the last output and termination.
If there are $L$ solutions, exactly $2L-1$ another-solution calls are made.
\end{theorem}

\paragraph{Constrained subproblems.}

Let $\mathit{In},\mathit{Ex}\subseteq V$ be disjoint sets.  The associated
solution family is
\[
    \mathcal F(\mathit{In},\mathit{Ex})
    \coloneqq
    \Set{X\in\I_1\cap\I_2\mid
        |X|=r,\ \mathit{In}\subseteq X,\
        X\cap\mathit{Ex}=\emptyset}.
\]
Whenever this subproblem is processed, the algorithm stores a known solution
$X\in\mathcal F(\mathit{In},\mathit{Ex})$.  The call
$\mathsf{Another}(\mathit{In},\mathit{Ex},X)$ asks for a different member of
this family and returns $\bot$ if none exists.

This constrained call is an ordinary Another Matroid Intersection instance
after contracting $\mathit{In}$ and restricting the ground set to exclude
$\mathit{Ex}$; equivalently, it uses the contracted and deleted matroids
\[
    (\M_i/\mathit{In})\setminus\mathit{Ex},
    \qquad i\in\Set{1,2},
\]
whose common ground set is
$V'=V\setminus(\mathit{In}\cup\mathit{Ex})$.  Because $\mathit{In}$ is a
subset of the stored common independent set $X$, an independence query for a
set $J\subseteq V'$ is implemented by the original query
\[
    J\text{ is independent in }
        (\M_i/\mathit{In})\setminus\mathit{Ex}
    \quad\Longleftrightarrow\quad
    J\cup\mathit{In}\in\I_i.
\]
The displayed solution is $X\setminus\mathit{In}$, and every returned solution
is lifted by adjoining $\mathit{In}$.  Thus one derived independence query
costs one original query.  This standard reduction of inclusion and exclusion
constraints to contraction and restriction (or deletion) is also used
in~\cite[Proposition~7]{KKW23}.

The enumeration procedure is given in \cref{alg:enumeration}.  In a call to
$\mathsf{RecEnumerate}$, its two displayed solutions $X$ and $Y$ have already
been output.  The solution $Z_0$ for the child excluding $e$ is found before
the child including $e$ is entered; it is stored without being output until
that first child has been exhausted.

\begin{algorithm}[H]
\caption{Enumeration with two-call worst-case delay}
\label{alg:enumeration}

\Procedure{\Enumerate{$X$}}{
    \textbf{output} $X$\;
    $Y\gets\Another(\emptyset,\emptyset,X)$\;
    \If{$Y\ne\bot$}{
        \textbf{output} $Y$\;
        $\RecEnumerate(\emptyset,\emptyset,X,Y)$\;
    }
}

\Procedure{\RecEnumerate{$\mathit{In},\mathit{Ex},X,Y$}}{
    choose $e\in X\setminus Y$\;
    $Z_0\gets\Another(\mathit{In},\mathit{Ex}+e,Y)$
        \Comment*[r]{look ahead; do not output $Z_0$ yet}
    $Z_1\gets\Another(\mathit{In}+e,\mathit{Ex},X)$\;
    \If{$Z_1\ne\bot$}{
        \textbf{output} $Z_1$\;
        $\RecEnumerate(\mathit{In}+e,\mathit{Ex},X,Z_1)$\;
    }
    \If{$Z_0\ne\bot$}{
        \textbf{output} $Z_0$\;
        $\RecEnumerate(\mathit{In},\mathit{Ex}+e,Y,Z_0)$\;
    }
}
\end{algorithm}

\begin{proof}[Proof of \cref{thm:enumeration}]
Consider a call
$\mathsf{RecEnumerate}(\mathit{In},\mathit{Ex},X,Y)$.  Since $X\ne Y$, an
element $e\in X\setminus Y$ exists.  Moreover,
\[
    X\in\mathcal F(\mathit{In}+e,\mathit{Ex}),
    \qquad
    Y\in\mathcal F(\mathit{In},\mathit{Ex}+e),
\]
and the two child families form the disjoint partition
\[
    \mathcal F(\mathit{In},\mathit{Ex})
    =\mathcal F(\mathit{In}+e,\mathit{Ex})
     \mathbin{\dot\cup}
     \mathcal F(\mathit{In},\mathit{Ex}+e).
\]
Thus both child calls have the stored solution required by
$\mathsf{Another}$.  Induction on the size of the solution family now shows
that \cref{alg:enumeration} outputs every maximum common independent set
exactly once.

Immediately after $Y$ is output, the recursive procedure makes exactly two
another-solution calls: it first tests the child that excludes $e$, and then
the child that includes $e$.  If $Z_1\ne\bot$, it is the next output.  If
$Z_1=\bot$ but $Z_0\ne\bot$, the stored solution $Z_0$ is the next output.
If both answers are $\bot$, the procedure returns.  After the recursion in the
first child finishes, $Z_0$ has already been computed, so the algorithm either
outputs it without another query or returns to its parent.  Consequently, a
chain of backtracking steps makes no queries: every two another-solution calls
produce a new output or certify that the remaining search is complete.  The
first output is given, and the call producing the second output uses only one
another-solution invocation.  This proves the delay bound, including the delay
after the final output.

Each output other than the initial solution starts exactly one invocation of
$\mathsf{RecEnumerate}$, and each such invocation makes two another-solution
calls.  Together with the initial call, their number is
$1+2(L-1)=2L-1$.  Every split fixes the membership of one previously unfixed
element, so the recursion depth is at most $n$.  Each active level stores only
the two constraint sets, a constant number of solutions, and one deferred
answer.  The space is therefore polynomial.
\end{proof}

\begin{corollary}\label{cor:enumeration}
All maximum common independent sets can be enumerated in polynomial space with
expected query delay $\widetilde O(n\sqrt r)$ by a Las Vegas algorithm, or
with deterministic query delay $\widetilde O(nr^{2/3})$.
\end{corollary}

\begin{proof}
Apply \cref{thm:enumeration} with the Las Vegas algorithm of
\cref{cor:randomized-matroid} or the deterministic algorithm of
\cref{cor:deterministic-matroid}.  Each constrained instance has at most $n$
elements and maximum solution size at most $r$.  Since every output history
determines the next constrained instance before the corresponding call is
made, the same argument applies conditionally to the expected query bound of
the Las Vegas algorithm.
\end{proof}

%----------------------------------------------------------------

\section*{Acknowledgments}

The author is grateful to Hanna Sumita, Tsubasa Harada, and Sotatsu Moriyama for helpful discussions and valuable comments on the manuscript.

%----------------------------------------------------------------

\section*{Declaration of Generative AI Use}

ChatGPT 5.6 sol was used to help generating the proof details and was also used to draft the manuscript. 
The algorithmic idea was contributed by the author. The author verified and revised the proof and the manuscript and takes full responsibility for the paper.

%----------------------------------------------------------------

\bibliographystyle{alpha}
\bibliography{references}

\appendix

\section{Concentration calculations}\label{app:concentration}

We record the probability calculations used in
\cref{lem:row-classification,lem:column-classification}.  The only
concentration result needed is the following standard form of Hoeffding's
inequality~\cite{Hoeffding63}.  If $X_1,\ldots,X_L$ are independent random
variables taking values in $[0,1]$, and
\[
    \overline X\coloneqq\frac{1}{L}\sum_{i=1}^L X_i,
    \qquad
    \mu\coloneqq\frac{1}{L}\sum_{i=1}^L\mathbb E[X_i],
\]
then, for every $t>0$,
\begin{equation}\label{eq:hoeffding}
    \Pr[\overline X-\mu\ge t]\le e^{-2Lt^2},
    \qquad
    \Pr[\mu-\overline X\ge t]\le e^{-2Lt^2}.
\end{equation}

\subsection{Row classification}
\label{app:row-classification-calculation}

\begin{proof}[Proof of \cref{lem:row-classification}]
Fix the active graph at the beginning of a phase and an implicit row $p$ of
degree $d\coloneqq d(p)$.  For $i\in\Set{1,\ldots,L}$, let $X_i$ be the
indicator that the $i$-th call to $\mathsf{HasIn}$ is positive.  The column
samples used in different experiments are independent, so the variables
$X_1,\ldots,X_L$ are independent Bernoulli random variables with common mean
\[
    \mu=1-\left(1-\frac{1}{4h}\right)^d.
\]

Suppose first that $d<2h$.  A union bound over the $d$ in-neighbors gives
\[
    \mu\le\frac{d}{4h}<\frac12.
\]
The row is labeled heavy only if $\overline X\ge3/4$.  Hence
\cref{eq:hoeffding} gives
\[
\begin{aligned}
    \Pr[p\text{ is labeled heavy}]
    &\le \exp\left(-2L\left(\frac34-\mu\right)^2\right)\\
    &\le \exp\left(-2L\left(\frac14\right)^2\right)
     =e^{-L/8}.
\end{aligned}
\]

Now suppose that $d\ge16h$.  Using $1-x\le e^{-x}$, we obtain
\[
    1-\mu
    =\left(1-\frac{1}{4h}\right)^d
    \le e^{-d/(4h)}
    \le e^{-4}.
\]
Thus $\mu\ge1-e^{-4}$.  The row is labeled light only if
$\overline X<3/4$, and therefore
\[
\begin{aligned}
    \Pr[p\text{ is labeled light}]
    &\le \exp\left(-2L\left(\mu-\frac34\right)^2\right)\\
    &\le \exp\left(-2L\left(\frac14-e^{-4}\right)^2\right)
     \le e^{-L/10},
\end{aligned}
\]
where the last inequality follows from
$2(\frac14-e^{-4})^2>1/10$.
\end{proof}

\subsection{Collective column classification}
\label{app:column-classification-calculation}

\begin{proof}[Proof of \cref{lem:column-classification}]
If $s=|\Aact|<4h$, every active column has degree less than $4h$ and the
procedure labels every implicit column light.  Thus both implications in the
lemma hold deterministically.  Assume below that $s\ge4h$, and fix an
implicit column $q$.  Let $N\subseteq\Aact$ be its active in-neighborhood and
write $d\coloneqq|N|=d(q)$.

For $i\in\Set{1,\ldots,L}$, let $X_i$ indicate that $q$ is detected in the
$i$-th experiment.  The sets $R$ sampled in different experiments are
independent, so $X_1,\ldots,X_L$ are independent Bernoulli random variables.
Write
\[
    k=\left\lfloor\frac{s}{4h}\right\rfloor,
    \qquad
    \mu=\Pr[R\cap N\ne\emptyset]=\mathbb E[X_i].
\]

If $d<2h$, then each row belongs to a uniformly random $k$-element set with
probability $k/s$.  A union bound over $N$ yields
\[
    \mu\le\frac{kd}{s}\le\frac{d}{4h}<\frac12.
\]
Exactly as in the first part of the row calculation,
\cref{eq:hoeffding} implies
\[
    \Pr[q\text{ is labeled heavy}]
    \le e^{-L/8}.
\]

Suppose instead that $d\ge16h$.  Since $s/(4h)\ge1$, we have
\[
    k=\left\lfloor\frac{s}{4h}\right\rfloor
    \ge\frac{s}{8h}.
\]
If $s-d<k$, a uniform $k$-element set cannot avoid $N$.  Otherwise, its
probability of avoiding $N$ satisfies
\[
    1-\mu
    =\frac{\binom{s-d}{k}}{\binom{s}{k}}
    =\prod_{j=0}^{k-1}\left(1-\frac{d}{s-j}\right)
    \le\left(1-\frac{d}{s}\right)^k
    \le e^{-dk/s}
    \le e^{-2}.
\]
Consequently, $\mu\ge1-e^{-2}$.  Applying the lower-tail bound in
\cref{eq:hoeffding} gives
\[
\begin{aligned}
    \Pr[q\text{ is labeled light}]
    &\le \exp\left(-2L\left(\mu-\frac34\right)^2\right)\\
    &\le \exp\left(-2L\left(\frac14-e^{-2}\right)^2\right)
     \le e^{-L/50},
\end{aligned}
\]
where the last inequality uses
$2(\frac14-e^{-2})^2>1/50$.
\end{proof}

%----------------------------------------------------------------

\end{document}